\documentclass{amsart}
\usepackage{graphicx}
\usepackage{amscd}
\usepackage{amsmath}
\usepackage{amsfonts}
\usepackage{amssymb}
\usepackage{bbm}
\usepackage{setspace}
\usepackage{enumerate}         
\usepackage{fixme}
\usepackage{color}
\usepackage{url}
\usepackage{amsthm}
\usepackage{bm}
\usepackage{xy}

\theoremstyle{plain}
\newtheorem{theorem}{Theorem}[section]

\newtheorem{algorithm}[theorem]{Algorithm}

\newtheorem{lemma}[theorem]{Lemma}

\newtheorem{proposition}[theorem]{Proposition}

\newtheorem{definition}[theorem]{Definition}

\theoremstyle{remark}
\newtheorem{remark}[theorem]{Remark}
\newtheorem{example}[theorem]{Example}

\numberwithin{equation}{section}

\newcommand{\ind}{1\!\kern-1pt \mathrm{I}}
\newcommand{\rsto}{]\!\kern-1.8pt ]}
\newcommand{\lsto}{[\!\kern-1.7pt [}

\numberwithin{equation}{section}

\newcommand{\RR}{\mathbb{R}}
\newcommand{\QQ}{\mathbb{Q}}

\newcommand{\EE}{\mathbb{E}}

\newcommand{\Esign}[1]{\mathbb{E}\left[ #1 \right]}   
\newcommand{\Econd}[2]{\mathbb{E}\left[\left.#1\right|#2\right]}        

\newcommand{\im}{\ensuremath{\mathsf{i}}}

\begin{document}
\title[Discrete time Term structure theory]{Discrete time Term structure theory and Consistent Recalibration Models}
\author{Anja Richter and Josef Teichmann}
\address{Baruch College, CUNY, 1 Bernard Baruch Way, New York, NY 10010, USA \\ ETH Z\"urich, D-Math, R\"amistrasse 101, CH-8092 Z\"urich, Switzerland}
\email{anja.richter@baruch.cuny.edu \\ jteichma@math.ethz.ch}
\thanks{The authors gratefully acknowledge the support from
ETH-foundation. Both authors are grateful for valuable comments from Ozan Akdogan, Philipp Harms and David Stefanovits. The second author is grateful for the hospitality of Chebychev Laboratory in St.~Petersburg, where a first version of this paper was presented, in particular for the valuable comments of Yana Belopolskaya and Elena Shmileva.}
\curraddr{}
\begin{abstract}
We develop theory and applications of forward characteristic processes in discrete time following a seminal paper of Jan Kallsen and Paul Kr\"uhner \cite{kk10}. Particular emphasis is placed on the dynamics of volatility surfaces which can be easily formulated and implemented from the chosen discrete point of view. In mathematical terms we provide an algorithmic answer to the following question: describe a rich, still tractable class of discrete time stochastic processes, whose marginal distributions are given at initial time and which are free of arbitrage. In terms of mathematical finance we can construct models with pre-described (implied) volatility surface and quite general volatility surface dynamics. In terms of the works of Rene Carmona and Sergey Nadtochiy \cite{cn09,cn12}, we analyze the dynamics of tangent affine models. We believe that the discrete approach due to its technical simplicity will be important in term structure modeling.
\end{abstract}

\subjclass[2010]{}
\date{\today}
\maketitle

\section{Introduction}

Model choice in finance is often done by calibrating model parameters (and initial values) to liquid derivatives' prices from today's market. This practice relies on the ad hoc class from which the model is chosen and on some selection criterion, i.e.~the solution of an inverse problem with respect to today's market prices. As a consequence the result comes with two labels: the ad hoc choice of the calibration procedure (model class and selection criterion) and the (sensitive) dependence on today's data. The difficulty with this approach is twofold: as soon as we enlarge the model class, important even well observable model properties might change, and second, when we \emph{re-calibrate} the model tomorrow the calibrated set of parameters might change. The first problem relates to non-robustness of the calibration procedure, the second to its inconsistency.


\medskip

Let us develop some terminology first: we shall always take a classical point of view, i.e.~we believe that discounted market prices (these might be an infinite vector, too) follow some martingale process $ {(S_t)}_{0 \leq t \leq T} $ over some sufficiently long time interval $ [0,T] $ (which should once and for all be seen as a finite set of discrete time points) with respect to one pricing measure $ \QQ $. When we consider derivatives, we consider a large financial market in which trading does not lead to arbitrages, whence the existence of a an equivalent martingale measure $ \QQ $, see \cite{fs04}. A model is a time-inhomogeneous Markov process $ {(X_t)}_{0 \leq t \leq T} $ together with some known deterministic function $ G $ such that $ S_t = G(t,X_t) $ for $ t \in [0,T] $. A model (class) comes with two sets of unknowns: first the current state $ X_0 $ of the model and second the model parameters determining its generator $ \mathcal{A} $. Very often current state and model parameters are addressed together as model parameters, but we carefully distinguish these entities: the stochastic movement of state variables $ X_t $ is a model feature, whereas the generator $ \mathcal{A} $ is a model constant. Apparently making the state space very high-dimensional helps to incorporate inconsistencies over time, but usually results in highly delicate calibration problems. This leads in, e.g.~interest rate theory, to the well-known HJM-theory, where consistency is a built-in model feature. On the other hand low-dimensional state spaces help to keep the calibration procedure feasible, stable and robust, but lead to model inconsistencies over time. This second approach corresponds to choosing local volatility or time-dependent L\'evy models as model class, where the state space is finally one-dimensional. By contrast in the HJM-inspired approach as for example in dynamic volatility surface (see \cite{sw08a}, \cite{sw08b}), local volatility (see \cite{cn09}) or tangent L\'evy models (see \cite{cn12}), the previously specified class of generators becomes a state variable, too.

\medskip

In this work we present a more balanced approach between both extremes, which relies on a careful re-visit of Hull-White extensions in a quite general setting. In words of interest rate theory: we want to keep model consistency in the sense of HJM-equations, but we also want to keep tractability in the sense of finite factor models. In order not to lose our readers' patience through tedious technical calculations we present a self-contained version of this theory in a discrete time setting. One can clearly see how the well-known difficulties of consistent dynamic volatility surface evolutions evaporate in the light of discrete time modeling. The first idea is standard: we set up an infinite dimensional state space which encodes every initial term structure by construction. As codebook we choose all possible distributional configurations without any restriction, in contrast to models where local volatility or tangent L\'evy processes are considered. The second idea is non-standard: we foliate the state space appropriately to obtain a balance between still rich re-calibration properties and satisfying dynamical properties. We outline this idea in the following paragraphs through a continuous time example borrowed from interest rate theory:

\medskip

Let $ {(R_t)}_{t \geq 0} $ be a Vasi\v{c}ek model for the short rate, i.e.
\[
dR_t = (b - a R_t) dt + \sigma dW_t \, , \; R_0 \in \mathbb{R} \, .
\]
Bond prices, which correspond to derivatives in this setting can be easily calculated, but today's (generic) bond prices $ T \mapsto P(0,T) $ can possibly \emph{not} be perfectly calibrated by the initial value $ R_0 $ and the three remaining model parameters $ a $, $b$, and $ \sigma $. The idea of the Hull-White extension is to introduce one time-dependent parameter in order to achieve perfect calibration, hence the following time-dependent Vasi\v{c}ek model
\[
dR_t=(b-aR_t) dt + \sigma dW_t + c(t) dt\, , \; R_0 \in \mathbb{R} \, 
\]
is suggested to replace the first time-homogeneous one. In this case we can calculate from $ T \mapsto P(0,T) $ -- under mild regularity assumptions and given parameters $ a $, $b$, $ \sigma $ and an initial value $ R_0 $ -- the functional form of $ t \mapsto c(t) $. However, the drawback is inconsistency over time, since tomorrow's re-calibration might lead to other model parameters $a$, $b$, $\sigma$ and in particular to another function $ t \mapsto c(t)$.

Looking again at the procedure of Hull-White extensions we actually see a two step methodology: first $ a $, $b$, $ \sigma $ and $ R_0 $ are fixed and then $ t \mapsto c(t) $ is calculated from initial values: in other words, in a first calibration step we approximately explain today's bond prices by choosing $ R_0 $ and model parameters $ a $, $ \sigma $ \emph{and a constant} $ b $. This will lead to quite poor results in calibration generically. In a second calibration step we choose a curve $ t \mapsto c(t) $ such that the model
\[
dR_t=(b-aR_t) dt + \sigma dW_t + c(t) dt\, , \; R_0 \in \mathbb{R} \, 
\]
explains the bond prices perfectly. This is by the previous considerations possible and leads to an operator $ \operatorname{C} $ mapping bond prices $ T \mapsto P(0,T) $, model parameters $ a $, $b $, $\sigma$ and state values $ R_0 $ to a curve $ t \mapsto c(t) $. There is an apparent redundancy in this procedure since $ b $ and $ t \mapsto c(t) $ have overlapping effects. We shall take an advantage of this redundancy: we can imagine a setting where in fact the parameters $ a $, $b$ and $ \sigma $ are stochastic, and the Hull-White extension is \emph{introduced} to guarantee time consistency.

Let us explain this in more detail: we allow now for stochastic changes of the parameters $ a $, $b$, $ \sigma $: consider an exogenously given process ${(\mathbf{a}_t)}_{t \geq 0} ={(a(t),b(t),\sigma(t))}_{t \geq 0} $, then we can in principle make sense of
\[
dR_t=(b-a(t) R_t) dt + \sigma(t) dW_t + \operatorname{C}(P^{\mathrm{tangent}}(t,.);a(t),b(t),R(t))(t) dt \, , \; R_0 \in \mathbb{R} \, , 
\]
and the initial bond prices are $ T \mapsto P(0,T) $. Here the notation $ P^{\mathrm{tangent}}(t,.) $ means the Vasi\v{c}ek bond prices at time $ t $ with parameters $ (a(t),b(t),\sigma(t)) $, an initial value $ R_t $ and Hull-White extension curve 
\[
s \mapsto c(s) = \operatorname{C}(P^{\mathrm{tangent}}(t,.);a(t),b(t),R(t))(s) \, , 
\]
for $ s \geq t $. From an analytic point of view the equation looks unclear, but we can quite easily imagine how one step of a splitting scheme of this equation would look like: we start with a (appropriately) small time step $ \Delta $, initial parameters $ a(0), b(0), \sigma(0) $ and $ R_0 $, an initial curve $ T \mapsto P(0,T) $ and a corresponding Hull-White extension 
$$ 
s \mapsto c(s) := \operatorname{C}(P^{\mathrm{tangent}}(0,.);a(0),b(0),R(0))(s)
$$
for $ s \geq 0 $. Next we simulate one Euler step into the future leading to $ R_\Delta $, i.e.
\[
R_\Delta = R_0 + (b(0)-a(0)R_0) \Delta + \sigma(0) dW_\Delta + c(0) \Delta
\] 
We then calculate the bond prices $ P^{\mathrm{tangent}}(\Delta,.) $ corresponding to $ a(0), b(0), \sigma(0) $, $ R_\Delta $, and $ s \mapsto c(s+\Delta) $. Now the second part of the splitting scheme starts: we simulate values $ a(\Delta) $, $b(\Delta)$, $ \sigma(\Delta) $ and calculate a new Hull-White extension $\tilde c := C(P^{\mathrm{tangent}}(\Delta,\cdot),a(\Delta),b(\Delta),R_\Delta)$.
With the resulting curve $ s \mapsto \tilde c(s)$ we restart the splitting step. Notice that this equation is \emph{not} an SDE in the usual sense, since $ \operatorname{C} $ is defined on the space of bond prices, respectively forward rates. It rather reminds a McKean-Vlasov equation by its characteristics at any time $t$ depending on the distribution of $ \int_t^. R_s ds $, for all times $t$. 

The model is initialized such that $ T \mapsto P(0,T) $ is perfectly calibrated, which provides also the initial Hull-White extension. It is clear that under certain mild regularity conditions on the parameter process $ \mathbf{a} $ and the initial term structure $ T \mapsto P(0,T) $ everything is well defined. It remains to understand consistency, i.e.
\[
\Esign{\exp \big( - \int_0^T R_s ds \big) } = P(0,T) \, .
\]
This, however, follows from the fact that
\[
\Econd{\exp \big( - \int_s^{t} R_u du \big) P^{\mathrm{tangent}}(t,T) }{\mathcal{F}_s}=P^{\mathrm{tangent}}(s,T)
\]
for $ 0 \leq s \leq t \leq T $, which seems to be the defining property of tangent term structure $ P^{\mathrm{tangent}}(t,T) $ in its integral description (the differential description being that \[ {(\exp \big( - \int_0^{t} R_u du \big) P^{\mathrm{tangent}}(t,T))}_{0 \leq t \leq T} \] is a martingale). However, the precise formulation of this theorem will be proved in the continuous time version of this paper. A discrete time version of it will be proved in the realm of the paper, see Theorem \ref{main_theorem}.

If -- instead of exogenous specification of $ \mathbf{a} $ -- we rely on day by day re-calibrations it is apparent that we can deal in this setting with changing model parameters without completely sacrificing the Vasi\v{c}ek dynamics. From an infinite dimensional point of view the infinitesimal dynamics still remains of Vasi\v{c}ek type, but we somehow change the type (more precisely the parameters $ a $, $b$ and $ \sigma $) of Vasi\v{c}ek dynamics at every instant.

One might ask what the actual advantage of considering consistent recalibration models instead of HJM type models is. The answer is two-fold: firstly, the increments of consistent recalibration models seen as HJM type models are of a special type, since they come from finite factor models, whence the models look like well-known finite factor models. In terms of numerical mathematics the distributional structure of increments can often be described exactly and need not be approximated on a stochastic basis.

As the previous example makes clear the foliation structure on the space of models, which is called Hull-White extension in the realm of interest rate models, plays a crucial role for this framework. We can now replace the Vasi\v{c}ek model by a general affine model and consider Hull-White extensions as specifications of leaves of foliations on the space of forward rates, see, e.g., \cite{ft04}. We call such models constructed via (generalized) Hull-White extensions \emph{consistent re-calibration models}. It is the purpose of this article to present the discrete time theory of such equations. This is one of the many places in mathematics, where an infinite dimensional view (in case of interest rates this is the HJM equation) on a finite dimensional equation (the short rate equation with re-calibration) helps to understand the theory.

After all these theoretical considerations the main result of this article can be described as follows: we can find stochastic processes $ {(\eta_t)}_{t \geq 0} $ taking values in a vector space of term structures (forward characteristics in the language of Section~\ref{forward-characteristics-section}), which have a linear structure of the type
\[
\eta_t = A_t + \sum_{i=1}^n B^i_t Y^i_t
\]
with respect to some factor driving process $ Y ={(Y_t)}_{t \geq 0} $, together with some underlying process $X$, to whom the term structure belongs (in the sense of forward characteristics of Section~\ref{forward-characteristics-section}). Already with deterministic coefficients $A$ and $B^i$ we can calibrate a sufficiently rich family of initial term structures $\eta_0$. Concatenating such processes consecutively leads to stochastic processes $\eta$ still describing arbitrage free evolutions of term structures, which can now also be re-calibrated according to the chosen modes of concatenation. Due to the linear model structure above, knowledge of a single trajectory of $ Y $ and $\eta$ allows to calculate $ B $ by quadratic variation estimators of $\eta$ and $Y$. Finally $A$ can be inferred by solving linear equations.

\medskip

We concentrate in this work first on the  multi-variate theory of forward characteristics, which will be introduced in Section \ref{forward-characteristics-section}. In Section \ref{applications-forward-characteristics} we outline several applications of forward characteristics in mathematical finance. In Section \ref{forward-characteristics-affine-processes} we show that the large model class of affine processes has particularly simple forward characteristics. In Section \ref{SDE-forward-characteristics} we introduce the corresponding HJM-type equation for forward characteristics and provide a solution theory via the classification of finite dimensional realizations, see Section \ref{fdr}, in Section \ref{crc}. These solutions correspond to consistent re-calibration models.


\section{Forward characteristic processes}\label{forward-characteristics-section}


Let $ (\Omega,\mathcal{F},(\mathcal{F}_t)_{t \geq 0},\QQ) $ denote a filtered probability space with augmented filtration over an infinite \emph{discrete set} of time points, which are denoted by abuse of notation by $t \geq 0$. Usually we think of integer multiples of a time tick, in particular we consider only equidistant time grids. We consider a (multi-variate) stochastic process $ X:={(X_t)}_{t \geq 0} $ taking values in $ \mathbb{R}^{n} $, which we can see, e.g., as logarithms of price processes or as integrated short rate process. All processes are assumed to be adapted, some will be predictable. Throughout the article the expectation under the (pricing) measure $\QQ$ will be denoted by $\EE$ and the scalar product on $\RR^n$ by $\langle \cdot , \cdot \rangle$. Further notation and details on discrete time stochastic finance can be found in \cite{fs04}.

We define first the term structure of marginal distributions of $ X $ in terms of their Fourier transforms.
\begin{definition}\label{forward-char}
Let $ X $ be an adapted stochastic process taking values in $ \mathbb{R}^n $, then a family of stochastic processes $ {\big(\eta_s(u,t)\big)}_{0 \leq s \leq t}$, $ u \in \mathbb{R}^n $, with $ \eta_s(0,t) = 0 $ for all $0 \leq s \leq t$ is called the process of \emph{forward characteristics of $X$} if
\[
\Econd{\exp(\im \langle u, X_t - X_s \rangle)}{\mathcal{F}_s} = \exp \big( \sum_{k=s}^{t-1} \eta_s(u,k) \big)
\]
for $ 0 \leq s \leq t $.
\end{definition}

\begin{remark}
Note that the normalization $ \eta_s(0,t) = 0 $ for all $0 \leq s \leq t$ ensures that the map $ u \mapsto \eta_s(u,t) $ is continuous and uniquely defined through the application of a complex logarithm. If we have additional exponential moments we can extend the definition of forward characteristics to an appropriate strip in $ \mathbb{C}^n $ containing purely imaginary vectors in its interior. Notice also that the definition implicitly assumes that the characteristic functions of increments $ X_t - X_s $ do not vanish, since $ \eta $ is always considered finitely valued. One could generalize here by allowing the value $ - \infty $, but we do not follow this path for the sake of simplicity.
\end{remark}

Forward characteristics generalize the concept of process characteristics, which -- in this discrete setting -- would simply correspond to the logarithm of the characteristic function of the increment $ X_{s+1} - X_s $, i.e.~the short end of the forward characteristic $ \eta_s(.,s) $, for $ s \geq 0 $.
\begin{definition}\label{process-char}
Let $ X $ be an adapted stochastic process taking values in $ \mathbb{R}^n $, then the family of adapted stochastic processes $ {(\kappa^X_s(u))}_{s \geq 0} $, $u \in \RR^n$, defined by
\[
\Econd{\exp \left(\im \langle u, X_{s+1} - X_s \rangle \right)}{\mathcal{F}_s} = \exp \big( \kappa^X_s(u) \big)
\]
and $ \kappa^X_s(0) = 0 $ for $ s \geq 0 $ is called \emph{(process) characteristic of $ X $}. 
\end{definition}

\begin{remark}
The process characteristic $(\kappa^X_s(u))_{s \geq 0}$, $u \in \RR^n$, of $X$ is uniquely defined through the property that the adapted stochastic process
\[
\left( \exp ( \im \langle u, X_s \rangle - \sum_{k=0}^{s-1} \kappa^X_k(u) )\right)_{s \geq 0}
\]
is a martingale for $u \in \RR^n$ and $ \kappa^X_s(0) = 0 $ for $ s \geq 0 $.
\end{remark}

Not every process qualifies as process of forward characteristics: we can prove the following consistency result, which characterizes forward characteristic processes. In order to state the result we shall assume a certain decomposition of the process $ \eta $ with respect to an additionally given process $ \varepsilon := {(\varepsilon_t)}_{t \geq 0} $, compare also \cite{kk10}.

\begin{definition}
We say that a complex-valued processes $ {(\eta_s(u,t))}_{0 \leq s \leq t}$ for $ u \in \mathbb{R}^n $ has a \emph{decomposition with respect to the $ \mathbb{R}^d $-valued process $ {(\varepsilon_t)}_{t \geq 0} $} if there are complex-valued, adapted processes $ {(\alpha_s(u,t))}_{0 \leq s < t}$ and $ {(\sigma^i_s(u,t))}_{0 \leq s < t}$ for $ u \in \mathbb{R}^n $, $ i=1,\ldots,d$, with 
\begin{equation}\label{normalization}
\alpha_t(u,t)=0 \, , \; \sigma^i_t(u,t) = 0
\end{equation}
for $ t \geq 0 $, $u \in \RR^n$ and $ i = 1,\ldots, d $
and such that
\begin{equation}\label{decomposition}
\eta_{s+1}(u,t)-\eta_s(u,t) = \alpha_s(u,t) + \sum_{i=1}^d \sigma^i_s(u,t) \Delta \varepsilon_s^i
\end{equation}
for $ 0 \leq s < t $ and $ u \in \mathbb{R}^n $. Here $ \Delta \varepsilon_t := \varepsilon_{t+1} - \varepsilon_t $, for $ t \geq 0 $.
\end{definition}

\begin{remark}\label{predictability-remark}
The normalization \eqref{normalization} is introduced to avoid the notion of predictability and to allow for simpler formulas. Note that these values can be chosen freely since they do not enter equation \eqref{decomposition}.
\end{remark}
\begin{proposition}\label{consistency}
Let $ X $ be an adapted stochastic process with values in $\mathbb R^n$ and let $ {(\eta_s(u,t))}_{0 \leq s \leq t}$ be a complex-valued family of adapted processes for $ u \in \mathbb{R}^n $ such that a decomposition \eqref{decomposition} exists with respect to a process $ \varepsilon $. Then $ \eta $ is the process of forward characteristics of $ X $ if and only if the \emph{short end condition} on the process characteristic of $ X $
\begin{equation}\label{short-end-condition}
\kappa^X_s(u)= \eta_s(u,s)
\end{equation}
for $ s \geq 0 $ and $ u \in \mathbb{R}^n $, and the \emph{drift condition}
\begin{equation}\label{drift-condition}
\kappa^X_s(u) - \sum_{k=s}^{t-1} \alpha_s(u,k) = \kappa^{(X,\varepsilon)}_s(u, - \im \sum_{k=s}^{t-1} \sigma^._s(u,k))
\end{equation}
for $ 0 \leq s \leq t $ and $u \in \RR^n$ hold true.
\end{proposition}

\begin{proof}
In contrast to the continuous time theory the proof of this statement is elementary since neither stochastic integration nor the theory of process characteristics is needed: let us assume that $ \eta $ is in fact the forward characteristic process of $ X $, then the short end condition is an immediate consequence of the respective definitions. It remains to be proved that the drift condition \eqref{drift-condition} holds true. This follows from the fact that
\[
\Econd{\exp(\im \langle u, X_t \rangle)}{\mathcal{F}_s} = \exp \big(\im\langle u, X_s \rangle  + \sum_{k=s}^{t-1} \eta_s(u,k) \big)
\]
for $ 0 \leq s \leq t $ and for $ u \in \mathbb{R}^n $ defines a martingale if and only if $ \eta $ is a process of forward characteristics (in particular $ \eta_s(0,t)=0 $ for $ 0 \leq s \leq t $). Hence we have to check whether
\[
\Econd{\exp \left(\im \langle u, X_{s+1} - X_s \rangle + \sum_{k=s+1}^{t-1} \big(\eta_{s+1}(u,k)-\eta_s(u,k)\big) - \eta_s(u,s) \right)}{\mathcal{F}_s} = 1
\]
for $ 0 \leq s < t $ and for $ u \in \mathbb{R}^n $. Let us assume martingality: when we insert the decomposition of $ \eta $ we obtain indeed the desired assertion, since
\[
\Econd{\exp\left(\im \langle u, X_{s+1} - X_s \rangle + \sum_{i=1}^d \Big( \sum_{k=s}^{t-1} \sigma_s^i(u,k) \Big) \Delta \varepsilon^i_s - \eta_s(u,s) + \sum_{k=s}^{t-1} \alpha_s(u,k)\right)}{\mathcal{F}_s} = 1
\]
if and only if
\begin{equation*}
\eta_s(u,s) - \sum_{k=s}^{t-1} \alpha_s(u,k) = \kappa^{(X,\varepsilon)}_s(u, - \im \sum_{k=s}^{t-1} \sigma^._s(u,k))
\end{equation*}
for $ 0 \leq s < t $ holds true by the very definition of the process characteristic $ \kappa^{(X,\varepsilon)} $ and the short end condition. Notice that we can extend the sum from $ k=s+1 $ to $k=s$ by Remark \ref{predictability-remark}. 

Assume conversely that the process $ \eta $ satisfies the two stated conditions, then by reversing the previous conclusions we obtain the desired martingale property, which concludes the proof.
\end{proof}

\begin{remark}
If one decomposes $ X $ in a similar way as $ \eta $, one can further simplify the previous conditions: we assume that
\begin{equation}\label{decomposition-X}
X_{s+1}-X_s  = \beta_s + \sum_{i=1}^d \gamma^i_s \Delta \varepsilon_s^i
\end{equation}
for $ 0 \leq s $ and $ u \in \mathbb{R}^n $, with adapted, $\mathbb{R}^n$-valued stochastic processes $ \beta $, $\gamma^i $, $i=1,\ldots,d$.
Then by Definition \ref{process-char} we have
\begin{equation}\label{process-char-Xdecomposed}
\kappa^{\varepsilon}_s(\langle u, \gamma^._s \rangle ) + \im \langle u , \beta_s \rangle  = \kappa^X_s(u)
\end{equation}
for $ s \geq 0 $ and $ u \in \mathbb{R}^n $, and
\begin{equation}\label{forward-char-Xdecomposed}
\kappa^{\varepsilon}_s(v+\langle u, \gamma^._s \rangle) + \im \langle u , \beta_s \rangle  = \kappa^{(X,\varepsilon)}_s(u,v)
\end{equation}
for $ s \geq 0 $ and $ u \in \mathbb{R}^n $, $ v \in \mathbb{R}^d $. This nicely explains how the joint process characteristic of $ (X,\varepsilon) $ expresses dependencies between $ X $ and $ \varepsilon $. Notice that there are no restrictions imposed on the process $ \varepsilon $ a priori, so that we could decompose the forward characteristics $\eta$ of $X$ with respect to $X$ itself -- if possible, i.e. $\varepsilon = X$. In the case of affine processes this is particularly useful, as can be seen in Chapter \ref{fdr}, where it allows for simpler formulae.
\end{remark}

\begin{remark}
The case of local independence is of particular importance (compare the corresponding notion in \cite{kk10}): $ X $ and $ \varepsilon $ are called \emph{locally independent} if
\begin{equation}\label{local independence}
\kappa^{(X,\varepsilon)}_s(u,v) = \kappa_s^X(u) + \kappa_s^\varepsilon(v)
\end{equation}
for $ u \in \mathbb{R}^n $, $ v \in \mathbb{R}^d $. In this case the drift condition \eqref{drift-condition} simplifies to
\begin{equation}\label{simplified-drift-condition}
- \sum_{k=s}^{t-1} \alpha_s(u,k) = \kappa^{\varepsilon}_s(- \im \sum_{k=s}^{t-1} \sigma^._s(u,k))
\end{equation}
for $0 \leq s \leq t $ and $ u \in \mathbb{R}^n $.
\end{remark}
\begin{remark}
Forward characteristics encode the term structure of distributions of increments of a stochastic process $ X $, i.e.~the distributions of $ X_t - X_s $, for $ 0 \leq s \leq t $, conditional on the information $ \mathcal{F}_s $ at time $ s $. Notice that there is redundant information in processes of forward characteristics (in contrast to processes of process characteristics), which in turn translates to drift conditions like \eqref{drift-condition} for processes of forward characteristics.
\end{remark}

\section{Applications of forward characteristics in mathematical finance}\label{applications-forward-characteristics}

In this section we introduce three examples, where forward characteristics actually appear in (well-known) models of mathematical finance from interest rate theory, option pricing theory and credit risk theory.

\subsection{Forward characteristics of a predictable process -- interest rate theory} We consider a discrete time bank account process
\[
B_s:=\exp(\sum_{k=0}^{s-1} R_k)
\]
for $ s \geq 0 $ with some real-valued short rate process $ {(R_s)}_{s \geq 0} $. We define the integrated short rate $ X_s := \sum_{k=0}^{s-1} R_k $ for $s \geq 0$ and consider its forward characteristics $ \eta $, i.e.
\[
\Econd{\exp(\im  u \sum_{k=s}^{t-1} R_k )}{\mathcal{F}_s} = \exp \big(\sum_{k=s}^{t-1} \eta_s(u,k) \big)
\]
for $ 0 \leq s \leq t $, which we assume to be decomposable with respect to a process $\varepsilon$. 

Since $ X $ is predictable, $X$ and any adapted process $ \varepsilon $ are locally independent, since
\[
\Econd{\exp(\im u( X_{s+1} - X_s ) + \im\langle v , \varepsilon_{s+1}-\varepsilon_s \rangle )}{\mathcal{F}_s} = \exp \big( \kappa^X_s(u) + \kappa^\varepsilon_s(v) \big)
\]
by the fact that $ X_{s+1} $ is $ \mathcal{F}_s $-measurable for $ s \geq 0 $.

Additionally, we assume that $ \EE( \exp((1+\delta) | X |)) < \infty $ for some $ \delta > 0 $, then we can extend the definition of forward characteristics to the strip $ \mathbb{R} \times {[-1,1]} \im $. 
If we choose $ u = \im $ in the equation above, we can identify the forward rate process in this discrete time setting with the process ${(-\eta_s(\im,t))}_{0 \leq s \leq t} $, more precisely the zero-coupon bond price $P(s,t)$ with maturity $t \geq s \geq 0$ is given by
\[
P(s,t) 
= 
\Econd{\exp(- \sum_{k=s}^{t-1} R_k )}{\mathcal{F}_s} = \exp \big(\sum_{k=s}^{t-1} \eta_s(\im,k) \big).
\]
Note that the drift condition \eqref{simplified-drift-condition} corresponds (in continuous time) to the famous HJM drift condition which was given in \cite{hjm90} and \cite{hjm92}, where the dynamics of forward rates were studied first time. In this sense forward characteristics and the corresponding drift condition extend the framework of forward rates and HJM-drift condition.

\subsection{Forward characteristics of a logarithm of a martingale -- option pricing theory} We consider the logarithm $ X $ of a martingale process $ S $ describing a discounted price process. Again we assume that $ \EE( \exp((1+\delta) | X |)) < \infty $ for some $ \delta > 0 $, then we can again extend the definition of forward characteristics to the strip $ \mathbb{R} \times {[-1,1]} \im $. The forward characteristics are related to prices of European options with ``Fourier'' payoff (see also the seminal paper \cite{cm99}) via the following formula
\[
\Econd{\exp(\im u X_t )}{\mathcal{F}_s} = \exp \big( \im u X_s + \sum_{k=s}^{t-1} \eta_s(u,k) \big)
\]
for $ 0 \leq s \leq t $. In other words: forward characteristics encode conditional laws of $ X_t $ given the information up to time $ s $, which amounts by Breeden-Litzenberger formulas to knowledge of the full option price surface. The dynamic, continuous time version of this approach is the starting point of the tangent L\'evy model approach \cite{cn12}.

The martingale condition for $ S = \exp(X) $ translates to $ \eta_s(- \im,t) = 0 $ for all $ 0 \leq s \leq t $, see corresponding formulas in \cite{kk10}. Local independence means absence of dependence of driving processes, or in terms of option pricing theory vanishing leverage.

\subsection{Forward characteristics of a predictable process and instantaneous recovery -- credit risk theory} We consider a two dimensional process $ (X^1,X^2) $ with predictable first coordinate and general second coordinate. Again we assume $ E( \exp((1+\delta) \| X \|)) < \infty $ for some $ \delta > 0 $ and extend the definition of forward characteristics to the strip $ \mathbb{R} \times {[-1,1]} \im $. From a point of view of credit risk theory we can understand $ S:=\exp(X^2) $ as the instantaneous recovery process and $ B:=\exp(X^1) $ as risk-free bank account process. We can define defaultable bond prices $ P(s,t) $
via
\[
S_s P(s,t) = \Econd{\frac{B_sS_t}{B_t}}{\mathcal{F}_s}
\]
for $0 \leq s \leq t$ and understand the evaluation of forward characteristics $ \eta $ at $ (\im,-\im) $ as logarithm of defaultable bond prices, with respect to some pricing measure $ \QQ $. More precisely,
\begin{align*}
P(s,t) = & \Econd{\frac{B_sS_t}{B_tS_s}}{\mathcal{F}_s} = \Econd{\exp\big(-(X^1_t-X^1_s)+(X^2_t-X^2_s)\big)}{\mathcal{F}_s} \\
= & \exp \big(\sum_{k=s}^{t-1}\eta_s(\im,-\im,t) \big) \, .
\end{align*}
Notice the analogy to foreign exchange markets, which follows ideas introduced in \cite{jartur:95}.

\section{Forward characteristics of affine processes}\label{forward-characteristics-affine-processes}

In this section we introduce discrete time affine processes and draw some elementary conclusions on their forward characteristics, which are particularly easy to calculate. For a rigorous treatment of continuous time-homogeneous and inhomogeneous affine processes see \cite{dfs03} and \cite{f05} respectively. As usual discrete time theory is completely elementary.

\begin{definition}\label{affine-process}
Let $ D $ be a closed, convex domain with non-empty interior containing $ 0 $ in $ \mathbb{R}^n $. A family of adapted stochastic processes $ (X_t^x)_{t \geq 0} $ for $ x \in D \subset \mathbb{R}^n $ is called a \emph{time-homogeneous affine process} if the affine property
\begin{equation}\label{affine property}
\Econd{\exp( \langle u, X_t^x \rangle)}{\mathcal{F}_s} = \exp \big(\phi( u,t-s) + \langle \psi(u,t-s), X_s^x \rangle  \big)
\end{equation}
for $ 0 \leq s \leq t $ with given deterministic function $ \phi, \, \psi^j : \mathcal{U} \times \mathbb{N} \to \mathbb{C} $, for $ j = 1,\ldots, n$, holds true. In other words: the conditional characteristic function is exponentially affine in the state vector $ X $. Here we denote by $ \mathcal{U} $ the union $ \cup_{m \geq 1} \mathcal{U}_m $ with
\[
\mathcal{U}_m := \{ u \in \mathbb{C}^n \; | \sup_{x \in D} |\exp(\langle u, x \rangle )| \leq m  \} \, .
\]
We shall always assume the normalization $ \phi(u,0) = 0 $ and $ \psi^j(u,0) = u $ for $ u \in \mathcal{U} $ and $ j=1,\ldots,d$, which makes the functions $ \phi $ and $ \psi $ unique.
\end{definition}

\begin{remark}
Affine processes are Markov processes since the conditional expectation on the whole past only depends on the present.
\end{remark}

\begin{remark}
We can analogously define time-inhomogeneous affine processes by requiring the existence of functions $ \phi(u,s,t) $ and $ \psi(u,s,t) $ for $ 0 \leq s \leq t $ and $ u \in \mathcal{U} $.
\end{remark}

\begin{proposition}
Let $ X $ be a time-inhomogeneous affine process, then the forward characteristics satisfy
\[
\sum_{k=s}^{t-1} \eta_s(u,k) = \phi(\im u,s,t) + \langle \psi(\im u,s,t) - \im u, X_s^x \rangle \, ,
\]
or -- by taking first differences --
\[
\eta_s(u,t) = \phi(\im u,s,t+1) - \phi(\im u,s,t) + \langle \psi(\im u,s,t+1) - \psi(\im u,s,t), X_s^x \rangle \, ,  
\]
for $ 0 \leq s \leq t $, respectively. The functions $ \phi $ and $ \psi $ are furthermore unique solutions of the following difference equations with \emph{vector fields} $ F $ and $ R $ (we use the notion of vector fields in analogy to Riccati ODEs),
\begin{align}\label{riccati-equations}
\phi(u,t,t+1) & =   F(u,t) \, \\ \nonumber
\psi(u,t,t+1)-u & =   R(u,t) \, \\ \nonumber
\phi(u,s,t+1) & =  F(u,t) + \phi(u+R(u,t),s,t) \, \\ \nonumber
\psi(u,s,t+1) & = \psi(u+R(u,t),s,t) \,
\end{align}
and initial values $ \phi(u,s,s)=0 $ and $ \psi(u,s,s) = u $ for $ 0 \leq s < t $.
\end{proposition}
\begin{proof}
The proof is a simple application of the affine property and uniqueness of logarithms, since
\begin{align*}
 & \Econd{\Econd{\exp( \langle u, X_t^x \rangle)}{\mathcal{F}_s}}{\mathcal{F}_r}  \\
 &= 
\Econd{\exp (\phi( u,s,t) + \langle \psi(u,s,t), X_s^x \rangle  )}{\mathcal{F}_r}  \\
&= \exp \big(\phi( u,s,t) + \phi(\psi(u,s,t),r,s) + \langle \psi(\psi(u,s,t),r,s) , X_r^x \rangle  \big )
\end{align*}
for $ 0 \leq r \leq s \leq t $, on the one hand. On the other hand the affine transform formula leads to
\[
\Econd{\Econd{\exp( \langle u, X_t^x \rangle)}{\mathcal{F}_s}}{\mathcal{F}_r}  = \exp \big(\phi( u,r,t) + \langle \psi(u,r,t), X_r^x  \rangle  \big) \, ,
\]
hence by uniqueness of the representation we obtain the semiflow property
\begin{align}\label{semiflow-property}
\phi(u,s,t) + \phi(\psi(u,s,t),r,s) & = \phi(u,r,t) \\
 \psi(\psi(u,s,t),r,s) = \psi(u,r,t) \nonumber
\end{align}
for $ 0 \leq r \leq s \leq t $, $ u \in \mathcal{U}$. These equations translate immediately to the asserted difference equations. On the other hand recursively defined solutions of the difference equation \eqref{riccati-equations} satisfy the semiflow property \eqref{semiflow-property}: we argue by induction in $ t-r =n $ assuming the statement is true for all $ 0 \leq r \leq s \leq t $ with $ t-r \leq n $. The Riccati equation \eqref{riccati-equations} then yields for $ 0 \leq r < s \leq t $
\[
\psi(u,r,t+1) = \psi(\psi(u,t,t+1),r,t) = \psi(\psi(\psi(u,t,t+1),s,t),r,s) = \psi(\psi(u,s,t+1),r,s),
\]
where the second and third equality use the induction hypothesis (notice that $ t+1 - s \leq n $). In an analogous manner we can conclude
\begin{align*}
\phi(u,r,t+1) & = \phi(u,t,t+1) + \phi(\psi(u,t,t+1),r,t) \\ 
& = \phi(u,t,t+1) + \phi(\psi(u,t,t+1),s,t) + \phi(\psi(\psi(u,t,t+1),s,t),r,s) \\
& = \phi(u,s,t+1) + \phi(\psi(u,s,t+1),r,s)
\end{align*}
for $ 0 \leq r < s \leq t $ assuming by induction hypothesis that the result holds for $ t - r \leq n $.
\end{proof}

\begin{remark}
For the purpose of this remark we switch from discrete to continuous time settings, where the theory of time-inhomogeneous affine process has been formulated in \cite{f05}. Notice that the Riccati difference equations look at first sight different from usual forward difference equations, since the ``vector fields'' are inserted in the flows and not as usual other way round. Taking continuous time limits this leads to the following transport PDEs
\begin{equation}\label{transport-PDE}
\frac{\partial}{\partial t} \psi(u,s,t) =  D \psi (u,s,t) (R(u,t))
\end{equation}
and
\[
\frac{\partial}{\partial t} \phi(u,s,t) =  F(u) + D \phi (u,s,t) (R(u,t)) \, 
\]
for $ 0 \leq s \leq t $ with standard initial conditions $ \psi(u,s,s) = u $ and $ \phi(u,s,s) = 0 $, which govern the structure of $ \phi $ and $ \psi $. However, such transport PDEs can be related to solutions of ODEs. Given a $C^1$-(time-inhomogeneous) solution flow satisfying generalized Riccati  ODEs
\[
\frac{\partial}{\partial s} \psi(u,s,t) =  - R(\psi(u,s,t),s) \text{ and } \frac{\partial}{\partial s} \phi(u,s,t) =  - F(\psi(u,s,t),s) \, ,
\]
for $ 0 \leq s \leq t $ with standard \emph{initial} conditions $ \psi(u,t,t) = u $ and $ \phi(u,t,t) = 0 $ (notice the backwards character in time $s$ of the ODEs). We obtain the semiflow properties $ \psi(\psi(u,s,t),r,s) = \psi(u,r,t) $ and $ \phi(u,s,t) + \phi(\psi(u,s,t),r,s) = \phi(u,r,t) $ for all $ 0 \leq r \leq s \leq t $. Differentiating these equations leads to the previously mentioned transport PDEs \eqref{transport-PDE}. In the case of time-homogeneous vector fields the difference between forward and backward flows is redundant since they coincide. Notice that for the theory of affine processes rather the forward transport equation, which corresponds to a backward ODE, is the conceptually correct point of view (which is, however, equivalent in continuous time).
\end{remark}

In applications of affine processes Hull-White extensions play a particular role, e.g.~in interest rate theory. We can consider Hull-White extensions for forward characteristics in a general affine framework. Let us first describe solution classes of Riccati equations \eqref{riccati-equations}, which correspond to changes in $F$:
\begin{proposition}
Consider the Riccati equation \eqref{riccati-equations} of a time-homogeneous affine process, i.e.~the vector fields $F$ and $R$ which characterize the system are not time-dependent. Consider a map $ (u,t) \mapsto \mu(u,t) $, then the Riccati equations associated to $ R $ and $ \tilde F(u,t) := F (u) + \mu(u,t) $ defined for $ t \geq 0 $ and $ u \in \mathcal{U} $ have a unique solution for all initial values and
\[
\widetilde{\phi}(u,s,t) = \sum_{k=s}^{t-1} F(\psi(u,t-1-k))+\mu(\psi(u,t-1-k),k) \text{ and } \widetilde{\psi}(u,s,t) = \psi(u,t-s)
\]
for $ 0 \leq s \leq t $ and $ u \in \mathcal{U} $.
\end{proposition} 

\begin{proof} By induction.
\end{proof}

The idea of Hull-White extensions is simply to see that by varying $F$ (without changing $R$) one can already obtain quite a variety of initial forward characteristics of processes on $D$, so called initial configurations. These configurations can be parametrized by an appropriate choice of a time-dependent map $\mu$. We need some notation to make this more precise: we apply here that complex logarithms of continuous functions $  \mathbb{R}^n \ni u \mapsto f(u) \neq 0 $ are uniquely defined if the value $ \log f(u) $ at $ u = 0 $ is fixed, see for instance \cite{mst13}. In our case we always use the normalization by $ 0 $ at $ u = 0 $.
\begin{definition}\label{Inc-set}
The set $ \operatorname{Inc}^D $ denotes the set of continuous functions vanishing at the points $ (0,t) $ 
\[
(u,t) \mapsto \log \EE \big[\exp ( \langle u ,\Delta Y_t \rangle)  \big]
\]
for $u \in \mathcal U$, $t \geq 0$ and for stochastic processes $Y$ taking values in $D$ such that the increments $\Delta Y_t$ satisfy $ y + \Delta Y_s \in D $ for all $y \in D$, and $ s \geq 0 $. 
\end{definition}

\begin{definition}\label{lying-above}
Fix an affine time-homogeneous process $X$ taking values in $D$ with characterizing functions $\phi$ and $\psi$ and initial value $x$ at time $0$. Let $ \mu \in \operatorname{Inc}^D $, then
\[
\widetilde{\phi}(u,s,t) = \sum_{k=s}^{t-1} F(\psi(u,t-1-k))+\mu(\psi(u,t-1-k),k) \text{ and } \widetilde{\psi}(u,s,t) = \psi(u,t-s)
\]
for $ 0 \leq s \leq t $ and $ u \in \mathcal{U} $, define logarithms of characteristic functions of distributions
\[
\sum_{k=0}^{t-1}\nu_0(u,k) := \widetilde{\phi}( \im u,0,t) + \langle \widetilde{\psi}(\im u,0,t) - \im u, x \rangle
\]
on $D$ for $ 0 \leq t $. 

We denote henceforward the set of such initial forward characteristics $ \nu_0 $ by $ I(x) $. Note that to every element $ \nu_0 \in I(x) $ there is at least one $ \mu \in \operatorname{Inc}^D $ generating it. 
In words: $I(x)$ is the set of \emph{initial configurations}, which lie ``above'' the marginal distributions of given affine process in the sense that $F$ is modified by an additional jump component, whose increment at time $t$ is characterized with $ \exp(\mu(u,t))$, for $ u \in \mathcal{U}$.
\end{definition}

The next proposition shows that for any initial configuration we can actually construct a time-inhomogeneous stochastic process attaining this forward characteristic at time $ s = 0 $:

\begin{proposition}\label{hull-white-extension}
Let $ X $ be a time-homogeneous affine process with characterizing functions $F$, $R$ and $\psi$. Then for every initial value $ x \in D $ and every initial configuration  $ \nu_0 \in I(x) $ there exists a unique stochastic process $ \widetilde{X} $ starting at $ x \in D $ with characterizing functions $ \widetilde{F} $, $R$ and $ \psi $, in the sense that  
\[
\sum_{k=s}^{t-1} \widetilde \eta_s(u,k) = \widetilde \phi(\im u,s,t) + \langle \psi( \im u, t-s) - \im u, \widetilde{X}_s \rangle  
\]
for $u \in \mathbb{R}^n$, $0 \leq s \leq t$, whose initial forward characteristic $\widetilde \eta_0 $ equals $ \nu_0 $. This process is called the \emph{Hull-White extension of $ X $ for a given initial configuration $ \nu_0 $}.
\end{proposition}

\begin{proof}
Let $ X $ be a time-homogeneous affine process as above starting in $x \in D$ and let $ \nu_0 \in I(x) $ be fixed. Then there exists $ \mu \in \operatorname{Inc}^D $ such that
\[
\sum_{k=0}^{t-1}\nu_0(u,k)=\widetilde{\phi}(\im u,0,t) + \langle \psi( \im u,t) - \im u, x \rangle
\]
for $ t \geq 0 $ and $ u \in \mathbb{R}^n$. Hence we can construct a process $\widetilde X$ on $D$ by setting (in law)
\[
\widetilde{X}_{t+1} = X_{1}^{\widetilde{X_t}} + \Delta Y_t,
\]
for $ t \geq 0 $. Here we use a stochastic process $ Y $ with increments $ \Delta Y_t $ independent of $ \mathcal{F}_t $ (on a possibly enlarged probability space), such that its increments satisfy
\[
E \big[ \exp( \langle u , \Delta Y_t \rangle ) \big] = \exp(\mu(u,t)) \, , \quad t \geq 0 \, ,
\]
and $ X_1^x$ a random variable realizing the time-homogeneous affine process at time $1$ independent of $ \mathcal{F}_t $ and $Y$, for $ x \in D$. This uniquely defines a stochastic process $ \widetilde{X} $ with state space $ D $ by assumption on $ X $ and $ \Delta Y $. The conditional characteristic function of $ \widetilde{X} $ is then deduced by iteration of one-step conditional expectations and the affine property of $X$:
\begin{align*}
& \Econd{\exp( \langle u, \widetilde{X}_t \rangle)}{\mathcal{F}_s} \\
= &\Econd{ \Econd{\exp( \langle u, X_1^{\widetilde{X}_{t-1}} + \Delta{Y}_{t-1} \rangle)}{\mathcal F_{t-1}}}{\mathcal{F}_s} \\
 = & \Econd{\exp( \phi(u,1) + \langle \psi(u,1) , \widetilde{X}_{t-1} \rangle + \mu(u,t-1))}{\mathcal{F}_s} \\
 = & \exp \big(\phi( u,t-s) + \langle \psi(u,t-s), \widetilde{X}_s \rangle + \sum_{k=1}^{t-s} \mu(\psi(u,k-1),t-k)) 
\end{align*}
for $u \in \mathcal U$ and $ 0 \leq s < t $. For the forward characteristics of $\widetilde X$ it follows
\begin{align*}
\sum_{k=s}^{t-1} \widetilde \eta_s(u,k) = \widetilde \phi( \im u,s,t)) + \langle \psi( \im u,t-s)) - \im u ,\widetilde X_s \rangle
\end{align*}
for $u \in \mathbb{R}^n$ and $0 \leq s \leq t$. This gives the result with $\widetilde{F}(u,t) := F(u) + \mu(u,t)  $ for $u \in \mathcal U$, $t \geq 0$.
\end{proof}

\begin{remark}
Notice that the process $ \widetilde{X} $ at this point is not a Markov process but only defined for initial value $ x \in D $ and the initial configuration $\nu_0$. However, we could define in a completely similar manner an affine, time-inhomogeneous Markov process with specifications $ \widetilde{F} $ and $ R $, which leaves $ D $ invariant.
\end{remark}
\begin{remark}
As a particular case of the above structure one can consider the situation, where only in the first coordinate, e.g., independent jumps are added.
\end{remark}
\begin{example}
We consider Hull-White extensions for a one dimensional affine process. Let $ (R_t)_{t \geq 0} $ be a Gaussian process satisfying the stochastic difference equation
\[
R_{t+1}=R_t+(b-aR_t)+\sigma W_t \, , R_0 \in \mathbb{R},
\]
for $ t \geq 0 $, real parameters $ a,b,\sigma $, and a sequence $ (W_t)_{t \geq 0} $ of independent, identically distributed centered normal random variables with variance $1$. This is a discrete time affine process since
\[
\Econd{\exp( u (R_{t+1} -R_t))}{\mathcal{F}_t} = \exp(u(b-aR_t)+\sigma^2 u^2/2) \, ,
\]
which yields $ R(u) = (1- a) u $ and $ F(u) = b u + \sigma^2 u^2/2 $. We can actually also calculate the functions $ \phi $ and $ \psi $, namely
\[
\psi(u,t-s)=(1-a)^{t-s}u
\]
and
\[
\phi(u,s,t)=\sum_{k=s}^{t-1} \big( b(1-a)^{t-1-k}u + \sigma^2 \frac{(1-a)^{2(t-1-k)}u^2}{2}\big) \, ,
\]
for $ u \in \mathcal U $. Hull-White extensions will change $ F $, which is in this case without constraints on the state space particularly easy. We start first with all possible Hull-White extensions, which are simply parametrized by a choice of a function $ \mu $, such that $ u \mapsto \mu(u,t) $ is a cumulant generating function for every $ t \geq 0 $. Then we define
\[
\widetilde{F}(u,t)= F(u) + \mu(u,t) \, ,
\]
which defines a time-inhomogeneous affine process with independent jumps of cumulant $ \mu(.,t) $ added at each point in time $t$. Calculating $ \widetilde{\phi} $ yields
\[
\widetilde{\phi}(u,s,t) = \phi(u,s,t) + \sum_{k=s}^{t-1} \mu((1-a)^{t-1-k}u,k)
\]
for $ u \in \mathcal U $. Given an initial forward configuration $ \nu_0 $, then an equation of the type
\[
\sum_{k=0}^{t-1}\nu_0(u,k)=\widetilde{\phi}(\im u,0,t) + \langle \psi( \im u,t)- \im u, x \rangle
\]
for $ t \geq 0 $ and $ u \in \mathbb{R}^n $ holds true, which allows to calculate recursively the function $ \mu $. We say that $ \nu_0 \in I(x) $ if and only if the function $ \mu \in \operatorname{Inc}^D $, i.e.~$\mu(.,t)$ is a cumulant generating function.
\end{example}


\section{A stochastic difference equation for forward characteristic processes}\label{SDE-forward-characteristics}


Affine processes, which appear as natural examples in the theory of forward characteristics, are also characterized by a distinguished geometric property. To state this property we need to formulate a stochastic difference equation for forward characteristics in order to construct Markov processes of forward characteristics.

The state space of the stochastic difference equation will be $ D \times \Theta^D $, but we investigate the equation on convex subsets $ \mathbb{R}^n \times \Theta^n $ for the purpose of convenience. $ \Theta^D $, or $ \Theta^n $, respectively, denotes the set of functions
$$
\mathbb R^n \times \mathbb N \ni (u,x) \mapsto \theta(u,x)
$$
such that $ \sum_{k=0}^{t-1} \theta(u,k) $ are logarithms of characteristic functions of stochastic processes taking values in $D$, or $ \mathbb{R}^n $, respectively. The formulation of the equation will however be on the Hilbert space $G$, given through the following definition:

\begin{definition}
Let $ G $ be a Hilbert space of continuous complex-valued functions defined on $ \mathbb{R}^n $ (or a more general set containing $ \mathbb{R}^n $ depending on the modeling purpose) , i.e. $ G \subset C(\mathbb{R}^n; \mathbb{C}) $. 

$ H $ is called a \emph{forward configuration Hilbert space} if $H$ is a Hilbert space of functions $ \theta: \mathbb{N} \to G $, i.e.~$H \subset l^2_w(\mathbb{N};G)$, i.e.~a weighted sequence space, such that 
\begin{enumerate}
\item we have a continuous embedding $ H \subset l^\infty_{\operatorname{loc}}(\mathbb{R}^n \times \mathbb{N}; \mathbb{C})$.
\item The shift semigroup $ (S_t \theta) (u,x) := \theta(u,t+x) $ acts as strongly continuous semigroup of linear operators on $ H $ for $ x,t \geq 0 $ and $u \in \mathbb R^n$.
\item Functions of finite activity L\'evy-Khintchine type
\[
 (u,t) \mapsto \im a(t)u - \frac{u^{T}b(t)u}{2} + \int_{\mathbb{R}^n} (\exp(\im u \xi) -1) \nu_t(d \xi) , \quad (u,t) \in \mathbb R^n \times \mathbb N,
\]
lie in $ H $, where $ a $, $ b $, $ \nu $ are functions defined on $ \mathbb{N} $ taking values in $ \mathbb{R}^n $, the positive-semidefinite matrices on $ \mathbb{R}^n $, and the finite positive measures on $ \mathbb{R}^n $, respectively (this corresponds to processes with independent increments and finite activity).
\end{enumerate}
\end{definition}

\begin{remark}
Notice that elements of the Hilbert space $H$ are understood in Musiela parametrization and therefore denoted by a different letter in the sequel. We have the relationship $ \eta_s(u,s+x) = \theta_s(u,x) $, with time to maturity  $x:= t-s$.
\end{remark}

In the sequel we are defining a stochastic difference equations, which express the conditions of Proposition \ref{consistency}:
\begin{definition}
Let $H$ be a forward configuration Hilbert space. We call the following system of stochastic difference equations
\begin{align}\label{sdiffe}
X_{t+1} - X_t & = \beta_t + \sum_{i=1}^d \gamma^i_t \Delta \varepsilon_t^i \, ,\\ \nonumber
\theta_{t+1} - \theta_t & = S_1 \theta_t - \theta_t  + \alpha(t,X_t,\theta_t) + \sum_{i=1}^d \sigma^i(t,X_t,\theta_t) \, \Delta \varepsilon^i_t \, \\ \nonumber
X_0 & \in \mathbb{R}^n \, , \; \theta_0 \in H \, ,
\end{align}
for $ t \geq 0 $ and maps 
\begin{align*}
\alpha &: \mathbb{N} \times \mathbb{R}^n \times H \to H \\
\sigma^i &: \mathbb{N} \times \mathbb{R}^n \times H \to H 
\end{align*}
a \emph{term structure equation for forward characteristics with vector fields $ (\alpha,\sigma) $, initial term structure $ \theta_0 $ and initial value $ X_0$}
\begin{itemize}
\item if $\varepsilon$ is a stochastic process taking values in $\mathbb R^d$ and $\beta$ and $\gamma^i$, $i = 1, \ldots, d$, are adapted stochastic processes with values in $\mathbb R^n$,
\item if, for $ s \geq 0 $, the following \emph{consistency condition} holds
\begin{equation}\label{X-char-equation}
\kappa^{\varepsilon}_s(\langle u, \gamma^._s \rangle ) + \im \langle u , \beta_s \rangle = \theta_s(u,0) \, ,
\end{equation}
\item and if, for $ s,x \geq 0 $ and $ (X,\theta) \in \mathbb{R}^n \times \Theta^n$, the following \emph{drift condition} is satisfied
\begin{align} \label{drift-condition2}
\kappa^{\varepsilon}_s(- \im \sum_{k=0}^{x} \sigma^.(s,X,\theta)(u,k) +\langle u,& \gamma^._s \rangle)
 = 
 \kappa^{\varepsilon}_s(\langle u, \gamma^._s \rangle ) - \sum_{k=0}^{x} \alpha(s,X,\theta)(u,k) \, . 
\end{align}
\end{itemize}
\end{definition}

\begin{remark}
Notice that the above stochastic difference equation is not a difference equation in the strong sense of the word as we require the conditions \eqref{X-char-equation} and \eqref{drift-condition2} additionally. Both conditions stem from Proposition \ref{consistency}, more precisely relations~\eqref{process-char-Xdecomposed} and ~\eqref{forward-char-Xdecomposed} respectively. The processes $\beta$ and $\gamma$ are hence determined from the short end of $\theta$.
\end{remark}

\begin{remark}
The introduction of the Musiela parametrization means in fact for the vector fields an additional shift by $1$ to the left. By abuse of notation we use the same letters as in Remark \ref{predictability-remark}, however, the evaluation at $ k=0 $ of the vector field $ \alpha(t,X,\theta)(u,0) $ corresponds to the evaluation at $ 1 $ in original coordinates. 
\end{remark}

\begin{remark}
We do not assume that $ \theta_t \in \Theta^n $, even though the interpretation of the stochastic difference equation might get lost at some point in time.
\end{remark}

\begin{remark}
If the vector fields $ \alpha $ and $ \sigma^i $ do not depend on $ X $ one can consider the stochastic difference equation for $ \theta $ on its own and construct $ X $ a posteriori. Usually one considers the $ \sigma^i $ a priori given and $ \alpha $ subject to a drift condition (which then also expresses the dependencies between $ X$ and $ \theta $ dynamics). We do not take this point of view, but rather choose vector fields $ \alpha $ and $ \sigma^i $ such that the drift condition is satisfied. So also $ \sigma^i$ might contain information on correlation or dependencies in certain parametric models, see Section \ref{fdr}.
\end{remark}
\begin{theorem}
Consider a term structure equation for forward characteristics \eqref{sdiffe} for initial values $ X_0 \in \mathbb R^n $ and $ \theta_0 \in \Theta^n $, assume that the $ \beta $ and $ \gamma $ are specified through \eqref{X-char-equation}, then the process $ X $ together with
\[
\eta_s(u,t):=\theta_s(u,t-s)
\]
for $ 0 \leq s \leq t $ is a process $X$ together with its forward characteristics $ \eta $.
\end{theorem}
\begin{proof}
The proof is a direct consequence of Proposition \ref{consistency}.
\end{proof}
\begin{remark}
Notice that we do not need to assume that $ \theta \in \Theta^n $, since this follows by the martingale condition directly. In other words: existence of solutions of the equation for \emph{all times} leads to stochastic invariance.
\end{remark}
\begin{remark}
The different possible choices for $ \beta $ and $ \gamma $ correspond to different dependence structures between $ X $ and $ \varepsilon $. Therefore uniqueness is not only not to be expected but also not desirable in the general case.
\end{remark}


\section{Finite dimensional Realizations}\label{fdr}


The theory of finite dimensional realizations, in contrast to its continuous time counterpart, is also substantially simpler. This is due to the fact that as soon as we are given a particular stochastic difference equation affine subspaces appear by the very nature of the difference equation. Hence stochastic invariance means immediately that affine subspaces lie inside the invariant sub-manifold. We shall also see that Hull-White extensions of affine models appear to be a natural class of examples of finite dimensional realizations. 

We start first with a basic consideration:
\begin{proposition}
Let $H$ be a Hilbert space and consider a stochastic difference equation with driving process $ \varepsilon $ of the type
\begin{equation}
\theta_{t+1} - \theta_t = S_1 \theta_t - \theta_t + \alpha(t,\theta_t) + \sum_{i=1}^d \sigma^i(t,\theta_t) \Delta \varepsilon^i_t \, , \; \theta_0 \in H,
\end{equation}\label{sdiffe-general}
with vector fields $ \alpha, \sigma^1,\ldots,\sigma^d : \mathbb{N} \times H \to H $. Assume that the support of each process increment $ \Delta \varepsilon_t $, $t \geq 0$, is full, i.e.~$\mathbb{R}^d$. 
Let $M \subset H $ be a $ k$-dimensional sub-manifold of $H$. The manifold $M$ is left invariant by the solutions of \eqref{sdiffe-general} starting in $M$, if and only if for all $ \theta \in M $ and $ \lambda \in \mathbb{R}^d $ it holds that
\[
\theta \mapsto S_1 \theta + \alpha(t,\theta) + \sum_{i=1}^d \sigma^i(t,\theta) \lambda^i \, .
\]
defines a map from $ M $ to $ M $. 

In particular $ M $ contains affine sub-spaces and along those subspaces
\[
\sigma^i(t,\theta) \in T_{S_1 \theta + \alpha(t,\theta)+\sum_{i=1}^d \sigma^i(t,\theta) \lambda^i} M
\]
for $ i = 1,\ldots,m $, $ \theta \in M $ and all $ \lambda \in \mathbb{R}^d$ holds true.
\end{proposition}
\begin{proof}
Applying invariance of $M$, the conditions on the increments and on their respective supports yields the assertion on the self map immediately. Taking derivatives with respect to $ \lambda $ yields the second assertion.
\end{proof}

At this point we do not investigate further the precise structure of such manifolds, which might be less interesting a task due to discrete time (compare \cite{ft04} in the continuous time setting, and the references therein). We just remark that those finite dimensional sub-manifolds, if they exist, are important since they constitute parametrized families of functions tailor-made for the model to calibrate initial forward configurations.

We know many examples of difference equations admitting finite dimensional invariant sub-manifolds in the realm of term structure equations for forward characteristics, for instance real components of affine stochastic volatility models:
\begin{definition}
An affine process $ (X,Y) $ with state space $ \mathbb{R}^n \times C \subset \mathbb{R}^{m+n} $, for a proper convex cone $ C $ is called \emph{affine stochastic volatility model} if
\begin{equation}\label{affine property-homogeneous}
\Econd{\exp( \langle u, X_t\rangle  + \langle v, Y_t\rangle)}{\mathcal{F}_s} = \exp \big(\phi(u,v,t-s) + \langle u, X_s \rangle + \langle \psi_C(u,v,t-s), Y_s \rangle  \big)
\end{equation}
for $ 0 \leq s \leq t $ with given deterministic function $ \phi,\psi_C^j : \mathcal{U} \times \mathbb{N} \to \mathbb{C} $, for $ j = 1,\ldots, m$, holds true. For convenience we drop the initial values $ (x,y) $ in the notation.
\end{definition}

\begin{remark}
Notice that those models also include discretizations of continuous time stochastic volatility models as long as the affine structure is preserved.
\end{remark}

The Hull-White extension of Proposition \ref{hull-white-extension} certainly applies to affine stochastic volatility models. It is a special feature of affine stochastic volatility models that the process $ Y $ is a Markov process in its own filtration. Therefore we can try to perform Hull-White extensions without changing the process $Y$ but only changing process characteristics for $X$. From a geometric point of view this leads to term structure equations for forward characteristics together with a foliation of finite dimensional sub-manifolds.
\begin{lemma}
Let $(X,Y)$ be an affine stochastic volatility model, then the cone-valued component $Y$ is a Markov process in its own filtration.
\end{lemma}

\begin{proof}
If we set $ u = 0 $ in equation \eqref{affine property-homogeneous} we see immediately that the right hand side only depends on $ Y $, which proves the Markov property.
\end{proof}

\begin{proposition}\label{hull-white-extension-stochvol}
Let $(X,Y)$ be an affine stochastic volatility model. Then for every initial value $ x \in \mathbb{R}^n $, $ y \in C$ and every initial configuration  $ \nu_0 \in I(x,y) =:I(y) $, which is defined by an cumulant function $ \mu \in \operatorname{Inc}^{\mathbb R^n} $ (i.e.~whose effect only acts on the first $n$ variables but not at all on $Y$), there exists a stochastic process $ (\widetilde{X},Y) $ starting at $ (x,y) \in \mathbb{R}^n \times C = D $ with characterizing functions $ \widetilde{F} $, $R$ and $ \psi $, in the sense that  
\[
\sum_{k=s}^{t-1} \widetilde \eta_s(u,v,k) = \widetilde \phi(\im u,\im v,s,t) + \langle \psi_C(\im u, \im v, t-s) - \im v, Y_s \rangle  
\]
for $(u,v) \in \mathbb{R}^{n+m}$, $0 \leq s \leq t$, whose initial forward characteristic $\widetilde \eta_0 $ equals $ \nu_0 $. This process is called the \emph{Hull-White extension $ (\widetilde{X},Y) $ of $ (X,Y) $ for a given initial configuration $ \nu_0 $}.
\end{proposition}

\begin{proof}
  Follows from Proposition \ref{hull-white-extension}.
\end{proof}

\begin{example} Let $(X,Y)$ be an affine stochastic volatility model, then the forward characteristic process of $ X $ is given through
\begin{equation} \label{forward-characteristics-affstochvol}
\sum_{k=s}^{t-1}\eta_s(u,k) = \phi (\im u,0,t-s) + \langle \psi_C( \im u,0,t-s) - \im u, Y_s \rangle
\end{equation}
for $ 0 \leq s \leq t $ and $ u \in \mathbb{R}^n $. Therefore we can define vector fields on the Hilbert space of forward configurations, namely
\[
\sigma^i(u,x)=\sigma^i(\theta)(u,x):= \psi^{i-n}_C( \im u,0,x+1)-\psi^{i-n}_C( \im u,0,x)
\]
for $ x \geq 0 $, $ u \in \mathbb{R}^n $ and $ \theta \in H $, $i=n+1,\ldots,n+m$. Choosing a driving process $ \varepsilon = (X,Y) $ (which need not necessarily be a martingale in all our considerations) we have a decomposition of $ \theta $ of the form
\begin{equation} \label{term-structure-equation-affstochvol}
\theta_{t+1}-\theta_t= S_1 \theta_t - \theta_t + \alpha(\theta_t) + \sum_{i=n+1}^{m+n} \sigma^i \Delta \varepsilon^i_t \, ,
\end{equation}
where $ \alpha $ is calculated according to the drift condition at the respective point $ \theta $. More precisely we can set $\alpha(\theta) = -S_1 \theta + \theta$. This can be seen by rewriting \eqref{forward-characteristics-affstochvol} in terms of $\theta$ which gives
\[
\sum_{k=0}^{t-s-1} \theta_s(u,k) = \phi(\im u,0,t-s) + \langle \psi_C(\im u,0, t-s) - \im u,Y_s\rangle
\]
for $ 0 \leq s \leq t $ and $ u \in \mathbb{R}^n $. Taking differences in $t$ and substituting $k=t-s$ we have 
\[
\theta_s(u,k) = \phi(\im u,0,k+1) - \phi(\im u,0,k) + \langle \psi_C(\im u,0, k+1)-\psi_C(\im u,0, k),Y_s\rangle
\]
for $ 0 \leq s $ and $ u \in \mathbb{R}^n $. It then follows that
\[
\theta_{t+1}-\theta_t = \sum_{i=n+1}^{m+n} \sigma^i \Delta Y^{i-n}_t 
\]
and hence \eqref{term-structure-equation-affstochvol} is satisfied with $\alpha(\theta) = -S_1 \theta + \theta$. Moreover the drift and consistency conditions are satisfied automatically since $\theta$ is defined as the forward characteristic process of $X$.
However, calibrating to an arbitrary initial term structure which does not apriori correspond to an affine stochastic volatility model requires a different choice of $\alpha$ and is not always possible. 
\end{example}


\section{Consistent Re-calibration Models}\label{crc}


The stochastic difference equation~\eqref{sdiffe} is due to its involved drift quite challenging. Fortunately the previous results yield a particularly simple method to solve a rich class of equations of type~\eqref{sdiffe}, namely models, whose one-step from $ t \mapsto t+1 $ is described by a time-inhomogeneous, affine stochastic volatility model (even with stochastically varying parameters). Since consistency conditions are fully expressed in one time step, concatenations of affine one-steps with different model parameters preserve consistency. However, one has to adapt the Hull-White extension after each time step. Notice also that concatenations of these models are in general not affine anymore, but still relatively easy to implement and calibrate:
\begin{definition}\label{crc-def}
Let $ \mathbf{a} $ be a parameter vector representing admissible parameters of an affine stochastic volatility model $ (X(\mathbf{a}),Y(\mathbf{a})) $ and consider an adapted process $ {(\mathbf{a}_t)}_{t \geq 0} $ taking values in the space of admissible parameters of the affine stochastic volatility model. 

Consider furthermore the set $ I(\mathbf{a}_t,y) $, which corresponds to the initial forward configurations $I(x,y)$ with admissible parameters $ \mathbf{a}_t$ from Definition \ref{lying-above}. For convenience we leave away the initial value $x$ (since it does not depend on $x$), but emphasize the dependence on the parameter vector $\mathbf{a}_t$ of the affine process. Hence the relation $ \theta \in I(\mathbf{a},y) $ means that there is at least one $ \mu \in \operatorname{Inc}^{\mathbb{R}^n} $ defining a Hull-White extension with initial forward characteristic $ \theta $ as in Proposition~\ref{hull-white-extension-stochvol}. Notice that this Hull-White extension is working with stochastic increments independent of $Y$.

We call a process $Z$ a \emph{consistent re-calibration model (CRC model)} if the system of equations
\begin{align}\label{concatenation}
Z^i_{t+1} - Z^i_t & = \Delta \varepsilon^i_t(\mathbf{a}_t) \, , \; i=1,\ldots,n \, ,\\ \nonumber
\theta_{t+1} - \theta_t & = S_1 \theta_t - \theta_t  + \alpha(\mathbf a_t) + \sum_{i=n+1}^{n+m} \sigma^i(\mathbf{a}_t) \, \Delta \varepsilon^i_t(\mathbf{a}_t) 
\end{align}
has a solution for some $ Z_0 \in \mathbb{R}^n$ and some $\theta_0 \in H $, for $ t \geq 0 $ in the set of forward characteristics.

We define the different coefficients and driving noises of the above equation: all processes are adapted and well defined on one stochastic basis. Let $ (\widetilde{X}(\mathbf a_t),Y(\mathbf a_t)) $ denote appropriate Hull-White extensions of a stochastic volatility model $(X(\mathbf a_t),Y(\mathbf a_t)) $ with parameters $\mathbf a_t$ in the sense of Proposition~\ref{hull-white-extension-stochvol} and with forward characteristic $ \theta_t $ at time $t$ (and, of course, starting at time $t$) on this stochastic basis. Hence in particular $ \theta_t \in I(\mathbf{a}_t,Y_t(\mathbf{a}_t))$, for $ t \geq 0$. We assume furthermore that 
$$ \widetilde{X}_{t+1}(\mathbf a_t)=\widetilde X_{t+1}(\mathbf a_{t+1}) \text{ and } Y_{t+1}(\mathbf a_t) = Y_{t+1}(\mathbf a_{t+1}) \, , $$ 
for $ t \geq 0 $, and that 
$$ \widetilde{X}_{t+1}(\mathbf a_t)- \widetilde X_{t}(\mathbf a_{t}) \text{ and } Y_{t+1}(\mathbf a_t) - Y_{t}(\mathbf a_{t}) \, , $$
is independent of 
$$ \widetilde{X}_{t}(\mathbf a_{t-1})-\widetilde X_{t-1}(\mathbf a_{t-1}) \text{ and } Y_{t}(\mathbf a_{t-1}) - Y_{t-1}(\mathbf a_{t-1}) \, , $$
given $ (\widetilde X_t(\mathbf a_t),Y_t(\mathbf a_t)) $, for $ t \geq 1 $.

For $t \geq 0$ let
\begin{align*}
\Delta \varepsilon_t(\mathbf{a}_t)&:= (\widetilde{X}_{t+1}(\mathbf{a}_t)-\widetilde{X}_t(\mathbf{a}_t),Y_{t+1}(\mathbf{a}_t)-Y_t(\mathbf{a}_t)), \\
\sigma^i(\mathbf{a}_t)(u,x)&:= \psi_C^{i-n,\mathbf{a}_t}(\im u,0,x+1)-\psi_C^{i-n,\mathbf{a}_t}(\im u,0,x) \, , 
\end{align*}
for $i = n+1, \ldots,n+m$, and 
\begin{align}\label{alpha}
& \alpha(\mathbf{a}_t)(u,x) :=  - \phi^{\mathbf{a}_t}(\im u,0,x+2)+2\phi^{\mathbf{a}_t}(\im u,0,x+1) - \phi^{\mathbf{a}_t}(\im u,0,x) + \\
& + \sum_{i=1}^{m}( - \psi_C^{i,\mathbf{a}_t}(\im u,0,x+2)+2\psi_C^{i,\mathbf{a}_t}(\im u,0,x+1) - \psi_C^{i,\mathbf{a}_t}(\im u,0,x)) Y^{i}_t(\mathbf{a}_t) \, .\nonumber
\end{align}
\end{definition}

\begin{remark}
A direct computation from Proposition~\ref{hull-white-extension-stochvol} would yield
\begin{align*}
& \alpha(\mathbf{a}_t)(u,x) :=  \widetilde{\phi}^{\mathbf{a}_t}(\im u,0,t+1,t+x+2) - \widetilde{\phi}^{\mathbf{a}_t}(\im u,0,t+1,t+x+1) - \\
& - \widetilde{\phi}^{\mathbf{a}_t}(\im u,0,t,t+x+2) + \widetilde{\phi}^{\mathbf{a}_t}(\im u,0,t,t+x+1) + \\
& + \sum_{i=1}^{m}( - \psi_C^{i,\mathbf{a}_t}(\im u,0,x+2)+2\psi_C^{i,\mathbf{a}_t}(\im u,0,x+1) - \psi_C^{i,\mathbf{a}_t}(\im u,0,x)) Y^{i}_t(\mathbf{a}_t) \, ,
\end{align*}
where $\widetilde{\phi}^{\mathbf{a}_t}$ solves the Riccati equation, started at $ t $, associated to some
$$ \widetilde{F}^{\mathbf{a}_t}(\theta_t)(u,v,s) =F^{\mathbf{a}_t}(u,v) +\mu(\theta_t)(u,s-t)  $$
derived from the fact that $ \theta_t \in I(\mathbf{a}_t,Y_t(\mathbf{a}_t))$. However, by
\[
\widetilde{\phi}^{\mathbf{a}_t}(u,v,t,t+x)=\phi^{\mathbf{a}_t}(u,v,x)+\sum_{k=0}^{x-1} \mu(u,t+k)
\]
we immediately obtain the above expression~\eqref{alpha} in Definition~\ref{crc-def}. The expression~\eqref{alpha} is remarkably simple, since only $ \phi $ and not $ \widetilde{\phi} $ appear therein. This is another justification for the use of $ (\widetilde{X},Y) $ as driving noises.
\end{remark}

\begin{remark}
Notice that the set $ I(\mathbf{a}_t,y) $ is independent of $x$ by \eqref{forward-characteristics-affstochvol}. Note that the drift term $ \alpha $ is calculated from the Hull-White-extension ``static'' affine stochastic volatility model provided in Proposition~\ref{hull-white-extension-stochvol}.
\end{remark}

\begin{remark}
The fact that process $Z$ of a CRC model has forward characteristics given by $ \theta $ follows from Proposition~\ref{hull-white-extension-stochvol} by induction: at each time $ t $ the next increment is just given by an ordinary Hull-White extension with parameters $ \mathbf{a}_t $ and initial configuration $ \theta_t $. By assumption $ \theta_{t+1} \in I(\mathbf{a}_{t+1},y) $ at $ t + 1 $ we know that another Hull-White extension for parameters $ \mathbf{a}_{t+1} $ and initial configuration $ \theta_{t+1} $ can be constructed.
\end{remark}

The following definition defines the set $ J(y,\theta) $ of parameters $ b $, for which -- given an initial value $ y $ -- the configuration $ \theta $ lies above.

\begin{definition}\label{J-set}
For a given affine stochastic volatility model $(X(\mathbf{a}),Y(\mathbf a))$ with $Y_0(\mathbf a)=y$ and forward characteristics $\theta$, the set $J(y,\theta)$ denotes the set of admissible parameters $ \mathbf{b} $ such that $ \theta \in I(\mathbf{b},y) $. 

Notice that $ J(Y_0(\mathbf a),\theta) $ contains at least $ \mathbf a $ if $ \theta \in I(\mathbf a,Y_0(\mathbf a)) $.
\end{definition}

\begin{theorem}\label{main_theorem}
Let $ (X(\mathbf a),Y(\mathbf a)) $ denote an affine stochastic volatility model with parameter vector $\mathbf a$. The previously introduced stochastic difference equation \eqref{concatenation} has solutions in law (defined on some possibly extended probability space) for an adapted process $ {(\mathbf{a}_t)}_{t \geq 0} $ taking values in admissible parameters if and only if
\begin{itemize}
\item the initial configuration $ \theta_0 $ lies above the affine stochastic volatility model with parameters $ \mathbf{a}_0 $, i.e.~$ \theta_0 \in I(\mathbf{a}_0,Y_0(\mathbf a_0)) $.
\item the parameter valued process satisfies $ \mathbf{a}_t \in J(Y_t(\mathbf a_{t-1}),\theta_t) $ for $ t \geq 1 $.
\end{itemize}
\end{theorem}
\begin{proof} By induction.
\end{proof}

\begin{remark}
CRC models models are concatenations of one step evolutions according to a Hull-White extension of an affine stochastic volatility process driven by an endogenously or exogenously given stochastic process $ (\mathbf{a}_t)_{t \geq 0} $. In other words: even though we are changing the parameters of the affine stochastic volatility model (which usually happens through re-calibration), we are still able to write a dynamics (using the technology of Hull-White extensions) which is consistent. From a numerical point of view, a well-chosen affine stochastic volatility model leads to stochastically well understood increments $\Delta \varepsilon$ which lead to lower complexity in simulation than pure HJM-type models. Furthermore, if the model parameters $\mathbf a$ do not change too quickly, pricing within a factor model on small time scales is possible. Hence CRC models are consistent in the long run and simultaneously incorporate the daily information appropriately. 
\end{remark}

It is remarkable that $ J(Y_t(\mathbf a_{t-1}),\theta_t) $ from Definition \ref{J-set} will often be quite a large set, even with non-empty interior, so no bad constraints are to be expected. Even though the admissible parameters change freely through the redundancy introduced by the Hull-White extension, we are able to buffer this effect.



\begin{algorithm}
The structure of an algorithm for simulation consequently looks as follows:
\begin{itemize}
\item Choose an initial term structure $ \theta_0 $, an initial vector of parameters $ \mathbf{a}_0 $, initial log-prices and variances $ X_0$, $Y_0 $ such that $ \theta_0 \in I(\mathbf{a}_0,Y_0) $.
\item Simulate one period of $ (X(\mathbf{a}_0),Y(\mathbf{a}_0)) $ with initial value $ X_0,Y_0 $ of the Hull-White extension with respect to initial forward configuration $ \theta_0 \in I(\mathbf{a}_0,Y_0) $.
\item The resulting configuration $ \theta_1 $ lies in $ I(\mathbf{a}_0,Y_1(\mathbf a_0)) $ by construction. Choose a random variable $\mathbf{a}_1$ such that $ \mathbf{a}_1 \in J(Y_1(\mathbf{a}_0),\theta_1)$ and continue by constructing a new Hull-White extension with respect to $ \mathbf{a}_1 $ for $ \theta_1 $.
\end{itemize}
\end{algorithm}

The result of Theorem~\ref{main_theorem} can be translated into the following time series situation: we have proved that there are stochastic processes $ {(\eta_t)}_{t \geq 0} $ of term structures of the form
\[
\eta_t = A_t + \sum_{i=1}^m B^i_t Y^i_t
\]
with respect to some factor driving process $ Y ={(Y_t)}_{t \geq 0} $, together with some underlying process $X$, to whom the term structure belongs in the sense of forward characteristics. In our construction there is an underlying affine stochastic volatility model which determines the construction. The coefficients $A$ and $B^i$ are stochastic and described by $ \widetilde \phi $ and $\psi^i_C$, for $i=1,\ldots,m$. The stochasticity of $B$ only depends on changes in the parameter vector $\mathbf a $, whereas the stochasticity in $A$ also depends on choosing the appropriate Hull-White extension. This process $\eta$ is describing sufficiently rich arbitrage free evolutions of term structures, which can be re-calibrated.

\begin{algorithm} Observations in this setting are given by a sufficiently long trajectory of $\eta$ and $X$, where $X$ allows to extract the time series of $Y$ (on a possibly coarser grid). A calibration algorithm could then look as follows:
\begin{itemize}
\item Choose a class of affine stochastic volatility models $ (X(\mathbf a),Y(\mathbf a)) $ parametrized by a parameter vector $ \mathbf a $.
\item Obtain the realized trajectory $ t \mapsto (X_t,\eta_t) $.
\item Extract (estimate) the realized trajectory $ t \mapsto Y_t $ with a non-parametric procedure.
\item Estimate from the time series of first differences of $\eta$ the (parametric) form of $ B^i $, which corresponds to determining $ \psi_C $, i.e.~determine all parameters of $\mathbf a$ appearing in $\psi_C$. 
\item Check whether $\eta_t$ still ``lies sufficiently above'', in order to justify the linear structure equation $\eta_t = A_t + \sum_{i=1}^m B^i_t Y^i_t$.
\item Determine $ A_t $ from the equation
\[
A_t = \eta_t - \sum_{i=1}^m B^i_t Y^i_t 
\]
and calibrate a full parameter vector $ \mathbf{a}_t $ from the term structure $\eta_t$, for $ t \geq 0 $.
\item Choose a model for $ t \mapsto \mathbf{a}_t $ obeying the consistency condition of Theorem \ref{main_theorem}.
\end{itemize}
\end{algorithm}

\begin{remark}
Instead of estimating with a two-step procedure model parameter processes and Hull-White extensions, we could also use Bayesian approaches, i.e.~filter model parameters for $ t \mapsto \mathbf{a}_t $. In contrast to classical calibration the previous calibration algorithm also has the ``Bayesian feature'' to build upon all information from the past.
\end{remark}

\begin{example}
This example is provided in continuous time for the reader's convenience but it can be transferred easily to discrete time by time discretization. Additionally it underlines that we can consider in our setting tangent affine models in the sense of \cite{cn12}, since we believe that the tangent model itself should already be as good as possible. In this example we describe the consistent re-calibration model based on the (tangent) Heston stochastic volatility model class from option pricing theory: Let $X$ be a continuous time Heston model for the log-price of an asset, i.e. 
\begin{align*}
 dX_t &= - \frac12 Y_t dt + \sqrt{Y_t} dW_t \\
 dY_t &= a(b-Y_t) dt +c \sqrt{Y_t} dB_t
\end{align*}
for $t \geq 0$. Here $W$ and $B$ are two Brownian motions with correlation $\rho \in [-1,1]$ and parameters $a,b,c \in \mathbb R$ satisfying the Feller condition $2ab \geq c^2$. As it is well-known, the parameter $b$ represents the long-run variance, $a$ is the rate at which the instantanous variance $Y$ approaches $b$ and $c$ represents the volatility of volatility.

We can derive the characterizing functions $F^{\mathbf a}$ and $R_C^{\mathbf a}$ (in continuous time!) of the Heston model $(X(\mathbf a), Y(\mathbf a))$, i.e.
\begin{align*}
F^{\mathbf a}(u,v) &= a b v \\
R_C^{\mathbf a}(u,v) &= \frac12 u^2 + \frac12 c^2 v^2 + c \rho u v - \frac12 u -(1+a) v
\end{align*}
for $(u,v) \in \mathcal U$ and the functions $\phi^{\mathbf a}$, $\psi_C^{\mathbf a}$ by solving the ODEs
\begin{align*}
\partial_t \phi^{\mathbf a} (u,v,t) &= F^{\mathbf a} (u ,\psi_C^{\mathbf a} (u,v,t) ) \\
\partial_t \psi_C^{\mathbf a} (u,v,t)  &= R_C^{\mathbf a} (u ,\psi_C^{\mathbf a} (u,v,t) )
\end{align*}
for $t \geq 0$ and with initial values $\phi^{\mathbf a }(u,v,0 ) = 0 $ and $\psi^{\mathbf a}_C(u,v,0 ) = v $. Unsurprisingly the Heston option prices in general do not fit all of today's options prices of different maturities well. We do not spell out the corresponding discrete time equations but we note that a continuous-time affine process observed at equidistant discrete times is a discrete time affine process.

First we explain the calibration procedure following the methods outlined in~\cite{ct:14}: given a trajectory (even discretely sampled!) of log prices $ X $, and of forward characteristics $ \eta $, we can infer the trajectory of $Y$ from the quadratic variation of $X$ and then infer the function $R$ and its correct parameter values $ t \mapsto (a(t),c(t),\rho(t)) $ of $R$ from the quadratic variation of $\eta$ by techniques of estimation of integrated quadratic variation of the continuous martingale part. Next we solve the defining linear equation for $A$
\[
A_t = \eta_t - \sum_{i=1}^m B^i_t Y^i_t 
\]
and obtain by the estimation of $B$ (which corresponds to $R$) and $ Y $ an expression for $ A $ along the observation time window. Both steps together provide us with a trajectory $ t \mapsto (a(t),b(t),c(t),\rho(t)) = \mathbf a_t $ and a choice of instantaneous Hull-White extensions $ t \mapsto \mu(\theta_t) $. This calibration therefore incorporates the full information of the time series of option prices and stock prices itself.

Finally we can select a model for $t \mapsto \mathbf a_t$ which satisfies $a_t \in J(Y_t(\mathbf a_{t-1}), \theta_{t})$.  This means, we choose $\mathbf a_t$ from those parameters $\mathbf a$, which have the property that the term structure $ \theta_t$ can be written as initial configuration $I(\mathbf b,Y_t(\mathbf a_{t-1}))$.

If instead of selecting a model based on data we want to simulate into the future, we choose an initial term structure $\theta_0$, initial log-price $X_0$ and initial variance $Y_0$ and choose a parameter vector $\mathbf a_0 = (a(0),b(0),c(0),\rho(0))$ such that $\theta_0 \in I(\mathbf a_0,Y_0)$. Then we add the Hull-White extension to the homogeneous model $(X(\mathbf a_0), Y(\mathbf a_0))$ and simulate this inhomogeneous model one step into the future. We can now choose a new admissible (random) parameter vector $\mathbf a_1  = (a(1),b(1),c(1),\rho(1))$ subject to the condition $\mathbf a_1 \in J(Y_1(\mathbf a_0), \theta_1)$, and continue as at the beginning of the paragraph resulting in an arbitrage-free evolution of forward characteristics calibrated to a time series.

\end{example}

\end{document}